\documentclass[11pt]{article}
\usepackage[margin=1in]{geometry} 

\usepackage{graphicx}
\usepackage{multirow} 
\usepackage{array,booktabs,arydshln,xcolor}

\usepackage{listings}
\usepackage{enumerate}
\usepackage{enumitem}
\usepackage{float}
\usepackage{tikz}
\usetikzlibrary{arrows}

\usepackage{caption}
\usepackage{mathtools}
\usepackage{hyperref}
\usepackage{fancyhdr}
\usepackage{relsize}

\usepackage{amsfonts}
\usepackage{amsmath}
\usepackage{amssymb}
\usepackage{bbm}
\usepackage{bm}
\usepackage{adjustbox}
\usepackage{ifthen}
\usepackage[ruled,linesnumbered]{algorithm2e}
\usepackage{amsthm, amssymb} 

\definecolor{toc}{RGB}{13,55,174}	
\usepackage{hyperref}				
\hypersetup{
    colorlinks=true,
    citecolor=red,
    filecolor=black,
    linkcolor=toc,
    urlcolor=toc
}

\allowdisplaybreaks 

\usepackage{thmtools}
\usepackage{thm-restate}

\newtheorem{theorem}{Theorem}[section]

\newtheorem{lemma}[theorem]{Lemma}
\newtheorem{corollary}[theorem]{Corollary}

\newtheorem{definition}{Definition}

\newcommand{\abs}[1]{\vert #1 \vert}

\newcommand{\eps}{\epsilon}

\newcommand{\tO}{\tilde{O}}

\newcommand{\R}{\mathbb{R}}
\newcommand{\C}{\mathcal{C}}
\newcommand{\levy}{L\'{e}vy}
\newcommand{\regcoeff}{\textsc{IrregCoeff}}

\newcommand{\of}{\overline{f}}
\newcommand{\ox}{\overline{x}}

\newcommand{\uf}{\underline{f}}
\newcommand{\ux}{\underline{x}}

\newcommand{\tF}{\tilde{F}}

\newcommand{\Fs}{F^{*}}
\newcommand{\fs}{f^{*}}

\date{} 

\title{Description Complexity of Regular Distributions}
\author{
Renato Paes Leme \\ Google Research \\ {\tt renatoppl@google.com} \and 
Balasubramanian Sivan \\ Google Research \\ {\tt balusivan@google.com} \and
Yifeng Teng \\ Google Research \\ {\tt yifengt@google.com} \and 
Pratik Worah \\ Google Research  \\ {\tt pworah@google.com}
}

\begin{document}

\maketitle

\begin{abstract}
Myerson's regularity condition of a distribution is a standard assumption in economics. In this paper, we study the complexity of describing a regular distribution within a small statistical distance. Our main result is that $\tilde{\Theta}{(\eps^{-0.5})}$ bits are necessary and sufficient to describe a regular distribution with support $[0,1]$ within $\eps$ \levy-distance. We prove this by showing that we can learn the regular distribution approximately with $\tilde{O}(\eps^{-0.5})$ queries to the cumulative density function. As a corollary, we show that the pricing query complexity to learn the class of regular distribution with support $[0,1]$ within $\eps$ \levy-distance is $\tilde{\Theta}{(\eps^{-2.5})}$. To learn the mixture of two regular distributions, $\tilde{\Theta}(\eps^{-3})$ pricing queries are required.
\end{abstract}

\section{Introduction}

A Myerson-regular distribution is a distribution with CDF $F$ such that the revenue curve in quantile space $R(q) = q \cdot F^{-1}(1-q)$ is concave. For distributions that can be represented by a PDF $f$ 
(i.e., have no point masses) this is equivalent to the (non-decreasing) monotonicity of the virtual value function 
\begin{equation}\label{eq:virtual_value}
\phi(v) = v - \frac{1-F(v)}{f(v)}.
\end{equation}

Myerson-regularity (or simply regularity) is a standard condition in Economics that was originally introduced by Myerson in his seminar paper on optimal auctions \cite{myerson1981optimal}. Since then it has played a fundamental role in the design and analysis of various economic setups such as:  bilateral trade \cite{myerson1983efficient}, prior-independent mechanism design \cite{bulow1994auctions,hartline2009simple,roughgarden2012supply,fu2013prior,sivan2013vickrey, allouah2020prior}, auctions from samples \cite{dhangwatnotai2010revenue, cole2014sample, fu2015randomization,huang2015making,morgenstern2015pseudo}, approximation in mechanism design  and revenue management \cite{celis2014buy, paes2016field, alaei2019optimal,golrezaei2021boosted}, ...

Besides being widely used in Economics, they also encompass various important distributions:

\begin{itemize}
\item distributions with log-concave PDF, i.e. distributions where the PDF is of the form $f(x) = \exp(-g(x))$ for a convex function $g$, such as uniform, exponential and normal.
\item distributions satisfying the monotone hazard rate conditions (MHR) condition which is a notion from reliability theory. If a random variable measures the time until a certain failure event happens (e.g. a light bulb goes out) then MHR means that the probability that the failure happens at any given moment conditioned it hasn't yet happened weakly increases over time.
\item equal-revenue distributions where pricing at each point of the support leads to the same revenue. As an example, consider the distribution supported on $[1, \infty)$ with CDF $F(x) = 1 - 1/x$. Those are known to be extremal regular distributions and are often used to derive lower bounds in revenue management.
\item certain distributions arising from machine learning, for example: let $x$ be a random feature vector in $\R^d$ which is sampled from a uniform or Gaussian distribution restricted to a $d$-dimensional convex set $K$ and let $w \in \R^d$ be a fixed vector of weights $w \in \R^d$. Then the distribution of the dot product $\langle w, x \rangle$ is regular by the Prekopa-Leidner Theorem \cite{prekopa1973logarithmic}.
\end{itemize}

\paragraph{Describing Regular Distributions} In this paper we ask how many bits of information we need to describe a regular distribution. We obtain a sharp bound and apply to derive new tight guarantees for learning regular distributions.

Since a regular distribution is a continuous object with possibly unbounded support, we need to make a few assumptions to make the problem well-defined.

\begin{itemize}
\item {\bf Bounded support}: We will assume that the distribution has support in $[0,1]$. Without this assumption, even representing the subclass of uniform distributions supported on $[n, n+1]$ requires infinitely many bits.
\item {\bf L\'{e}vy-distance approximation} We will allow for $\epsilon$-error in the representation measured in L\'{e}vy distance, which allows an $\epsilon$ error both in values and probabilities. Formally:
\begin{equation}\label{eq:levy}
 \text{\levy}(F,G) := \inf\{ \epsilon \text{ s.t. } F(v-\epsilon) - \epsilon \leq G(v)  \leq F(v+\epsilon) + \epsilon, \forall v \}
 \end{equation}
\end{itemize}

Now we can state our main question formally:

\begin{definition}\label{dfn:dc}
We say that it is possible to describe a class of distributions $\C$ with 
 $b$ bits and $\epsilon$ error, if there is a class of distributions $\C'$ with $\abs{\C'} \leq 2^b$ such that for every distribution $F \in \C$ there exists $F' \in \C'$ such that $ \text{\levy}(F,F') \leq \epsilon$.
\end{definition}

Our main theorem is:

\begin{theorem}\label{thm:main}
It is possible to describe the class of regular distributions bounded in $[0,1]$ within $\epsilon$ \levy-distance using $\tilde{O}(1/\epsilon^{0.5})$ bits. Moreover, this is tight up to polylog factors.
\end{theorem}

It is useful to contrast this with the class of \emph{general} distribution supported in $[0,1]$ for which $\Omega(1/\epsilon)$ bits are necessary and $\tilde{O}(1/\epsilon)$ bits sufficient. For the sufficient part, we can represent a CDF $F$ by the numbers $\lfloor F(k\epsilon) / \epsilon \rfloor$ for $k = 0, 1, \hdots, \lceil 1/\epsilon \rceil$. Given those, we can construct a distribution with CDF:
$$\hat{F}(x) = \epsilon \lfloor F(k\epsilon) / \epsilon \rfloor \text{ for } x \in [k\epsilon, (1+k) \epsilon)$$
which is $\epsilon$-close to $F$ in \levy-distance. To see that $\Omega(1/\epsilon)$ bits are necessary, it is enough to construct a set of $2^{\Omega(1/\epsilon)}$ distributions such that each pair differ by at least $\epsilon$ in \levy-distance. Given bits $b_0, \hdots, b_k \in \{0,1\}$ for $k = 1/(2\epsilon)$ construct a distribution such that:
$$F_b(x) = 2(k + b_k)\epsilon \text{ for } x \in [2k\epsilon, 2(1+k) \epsilon), k=0, 1, \hdots, 1/(2\epsilon)$$

For regular distributions, we will be able to construct a more succinct representation, using only square root of the number of bits. However, we will need a more sophisticated sampling procedure.

\paragraph{Beyond Regular Distributions}

While regular distributions are common in auction theory, many distributions encountered in practice are irregular. To generalize our results beyond regularity we define a notion of an irregularity coefficient, which measures how close to regular a distribution is:

\begin{equation}\label{eq:regcoeff}
\regcoeff(F) = \inf \left\{ \beta \geq 0 \text{ s.t. } q F^{-1}(1-q) + \beta \int_{1-q}^1 F^{-1}(x) dx \text{ is concave}\right\}
\end{equation}

A $0$-irregular distribution is a regular distribution in the usual sense and an $\infty$-irregular distribution is a general distribution. The $\beta$-irregular class for $\beta>0$ contains irregular distributions that are close enough to regularity to afford a low description complexity:

\begin{theorem}\label{thm:beta-irregular}
It is possible to describe the class of $O(1)$-irregular distributions bounded in $[0,1]$ within $\epsilon$ \levy-distance using $\tilde{O}(1/\epsilon^{0.5})$ bits.
\end{theorem}

The notion of $\beta$-irregularity coincides with the notion of $\alpha$-strong regularity of Cole and Roughgarden \cite{cole2014sample} when the sign is negative. The definition was originally intended to interpolate between regular and MHR distributions for positive values of $\alpha$. When one considers negative values for $\alpha$, we obtain a measure of how close a certain distribution is to regularity. A distribution is $\alpha$-strongly-regular in the sense of \cite{cole2014sample} whenever:

$$
f'(v) \cdot (1-F(v)) \geq- (2-\alpha) f^2(v)
$$
and hence it is $\beta$-irregular when it is $(-\beta)$-strongly-regular, i.e.:
\begin{equation}\label{eq:alpha-reg}
f'(v) \cdot (1-F(v)) \geq- (2+\beta) f^2(v)
\end{equation}

\paragraph{Application: Pricing Query Complexity}

Our first application is to settle the pricing query complexity of learning a regular distribution with $\epsilon$-error in \levy-distance. The pricing query complexity model was introduced in \cite{leme2021pricing} with the goal of learning in economic settings where the only viable mechanism is posted prices. In such settings, we only observe if prices posted for different agents led to a sale or no sale. This notion is also useful when the auction of choice is a first-price auction, since we don't have access to truthful bids, but we still know if the bidder chose to bid above the reserve or not. In such settings, we use the binary sale/no-sale outcomes observed from previous periods to optimize the price in future auctions.

In this model a learning algorithm is able to interact with a distribution $F$ via pricing queries: in each query, the algorithm chooses a price $p_t$, a sample $v_t$ is drawn from $F$ and the algorithm only learns the sign $\{+1, -1\}$ of $v_t-p_t$. The goal of the algorithm is to estimate some parameters of the distribution such as the mean, median or monopoly price for a given class of distributions. Paes Leme, Sivan, Teng and Worah \cite{leme2021pricing} give matching upper and lower bounds for several parameters of interest, but the leave a gap in estimating the pricing query complexity of learning the CDF of a regular distribution. They give provide a $\tilde{O}(1/\epsilon^3)$ upper bound and a $\Omega(1/\epsilon^{2.5})$ lower bound.

We settle this question by providing a $\tilde{O}(1/\epsilon^{2.5})$, matching the lower bound up to polylogarithmic factors. This is obtained by computing the $O(1/\epsilon^{0.5})$ description using $O(1/\epsilon^{2})$ pricing queries to evaluate each point using the Chernoff bound:

\begin{theorem}\label{thm:pricing-query-regular}
There is a $\tilde{O}(1/\epsilon^{2.5})$ upper bound on the pricing query complexity of learning the CDF of a regular distribution within $\epsilon$ \levy-distance error. This result is tight up to polylogarithmic factors.
\end{theorem}

\paragraph{Application: Mixture Distributions}

We say that a distribution with CDF $F$ is a mixture of $k$ distributions $F_1, \hdots, F_k$ if there are weights $w_1, \hdots, w_k \geq 0$ with $\sum_i w_i = 1$ such that $F(v) = \sum_i w_i F_i(v), \forall v$.
In various applications, it is useful to write distributions as mixtures of other distributions. For example, Sivan and Syrgkanis \cite{sivan2013vickrey} design auctions for mixtures of regular distributions whose performance depends on the number of distributions in the mixture. One may ask how many regular distributions are needed to represent a general distribution. The description complexity bounds automatically imply the following corollary:

\begin{corollary}\label{cor:mix1} There exist a irregular distribution that can't be $\epsilon$-approximated in \levy-distance by a mixture of $o(1/\epsilon^{0.5})$ regular distributions.
\end{corollary}

This follows directly from the fact that a mixture of $k$ regular distributions can be described by $O(k/ \epsilon^{0.5})$ while a general distribution requires $\Omega(1/\epsilon)$ bits to represent.

It is also useful in ML applications to represent distributions as mixtures of Gaussians. Our result also implies that one may require many Gaussians to represent a regular distribution, since a Gaussian distribution can be represented by $\tilde{O}(1)$ bits.

\begin{corollary}\label{cor:mix2} There exist a regular distribution that can't be $\epsilon$-approximated in \levy-distance by a mixture of $o(1/\epsilon^{0.5})$ Gaussian distributions.
\end{corollary}

\paragraph{Learning Mixtures of Regular Distributions}

A mixture of two regular distributions can be described using $\tilde{O}(1/\epsilon^{0.5})$ bits since we need $\tilde{O}(1/\epsilon^{0.5})$ bits to describe each distribution and $\tilde{O}(1)$ bits more to describe the weights. Given that, one would guess that the pricing query complexity of learning a mixture is also $\tilde{O}(1/\epsilon^{2.5})$. We conclude the paper with the following rather surprising result:

\begin{theorem}\label{thm:learning_mixture_regular}
Let $\C$ be the class of distributions supported in $[0,1]$ that can be written as a mixture of two regular distributions. To estimate the CDF of a distribution in that class within $\epsilon$ in \levy-distance we require at least $\Omega(1/\epsilon^3)$ pricing queries.
\end{theorem}

\paragraph{Why Levy Distance?} There are different ways to measure the distance between two distributions such as Total Variation (TV), Wasserstein, Kolmogorov, and \levy. We note that bounds on the \levy-distance automatically imply bounds on the Wasserstein. The Kolmogorov and TV distances are stronger notions but it is impossible to obtain any approximation in either one using pricing queries. This is discussed in detail in \cite{leme2021pricing}. The Kolmogorov distance between two distributions $F$ and $G$ is given by:
\begin{equation}\label{eq:kol}
\text{Kol}(F,G) := \inf \{ \epsilon \text{ s.t. } F(v) - \epsilon \leq G(v) \leq F(v) + \epsilon, \forall v \}
 \end{equation}
If a distribution is a deterministic value in $[0,1]$, to get any meaningful approximation in Kolmogorov distance, one needs to estimate this value exactly, which is impossible using pricing queries. 

Another reason to study \levy-distance is that such a metric has been considered in the literature on the sample complexity of learning revenue-optimal auctions, see Brustle et al \cite{brustle2020multi} and Cherapanamjeri et al \cite{cherapanamjeri2022estimation} for examples. Better understanding on the complexity of estimating a distribution within \levy-distance can lead to improved results on the complexity of auction learning.

\paragraph{Related Work} Our work is broadly situated in the theme of sample complexity in algorithmic economics, where the goal is to understand what is the minimal amount of information to describe or optimize a certain economic setup. There are several ways one can explore this question. For example, Dhangwatnotai et al \cite{dhangwatnotai2010revenue} and Fu et al \cite{fu2013prior} ask to what extent one can optimize an auction using a single sample of a distribution. In the other extreme, we can ask how many samples from a certain distribution are required to optimize an auction (Cole and Roughgarden \cite{cole2014sample} and Morgenstern and Roughgarden \cite{morgenstern2015pseudo}) or compute the optimal reserve price (Huang et al \cite{huang2015making}).

Closest to our paper is the paper by Paes Leme at al \cite{leme2021pricing} which considers a restricted query model. Instead of having access to samples of a distribution, we are only allowed to post a price and observe for a fresh draw of that distribution whether the price was above or below the posted price. This is motivated by learning in two important scenarios: (i) settings where posted-price mechanisms are used and we only observe purchase/no-purchase decisions; (ii) learning in non-truthful auctions (like first-price auctions) where the bid is not an unbiased sample of the value but we can still observe whether a bidder decided to bid above the reserve price or not. In this setting, Paes Leme at al \cite{leme2021pricing} provides tight bounds on how to learn the monopoly price of MHR, regular and general distributions and shows that it requires strictly less samples to learn the reserve than to learn the entire CDF of a regular distribution.

However, \cite{leme2021pricing} leaves a gap on the number of pricing queries required to learn the CDF of a regular distribution: they show a lower bound of $\Omega(1/\epsilon^{2.5})$ and an upper bound of $\tilde{O}(1/\epsilon^3)$. We settle this question by providing an algorithm with pricing query complexity $\tilde{O}(1/\epsilon^{2.5})$. 

Our work is also related to the line of work on query complexity in Computer Science, which asks what is the minimum number of queries to a black box required to perform a certain task. This approach has been applied to learning theory, parallel computing, quantum computing, analysis of Boolean functions, and optimization, among others. Since this literature is too broad to be surveyed here, we refer the reader to the book by Kothari, Lee, Newman, and Szegedy \cite{szegedy2023query}.

\section{Description Complexity of Regular Distributions}

An $\Omega(1/\epsilon^{0.5})$ lower bound on the description complexity of regular distributions is implicit in the $\Omega(1/\epsilon^{2.5})$ lower bound on the pricing query complexity given in \cite{leme2021pricing}. We can even modify  their example to strengthen the result to work for the class of distributions with non-decreasing PDF functions (which is a subclass of MHR distributions and regular distributions).

\begin{theorem}
\label{thm:description-lb}
There is a set of $2^{\Omega(1/\epsilon^{0.5})}$ distributions  with non-decreasing PDF supported on $[0,1]$ such that for every given pair the \levy-distance is at least $\Omega(\epsilon)$.
\end{theorem}

\begin{proof} Let $n = 1/(4 \epsilon^{0.5})$ and let $b = (b_0, \hdots, b_{n-1}) \in \{0,1\}^n$ be a sequence of bits. Now define a PDF $f_b(x)$ such that for $v \in [4\sqrt{\epsilon}k, 4\sqrt{\epsilon}(k+1)]$ we have $f_b(v) = 2v$ if $b_k = 0$ and 

\begin{equation*}
     f_b(v) = \left\{\begin{array}{lr}
        2v, & \text{if } v\in[v_0-2\sqrt{\eps},v_0-\sqrt{\epsilon}]\cup[v_0+\sqrt{\eps},v_0+2\sqrt{\epsilon}];\\
        2v_0, & \text{if } v\in(v_0-\sqrt{\epsilon},v_0+\sqrt{\epsilon}).\\
        \end{array}\right. \qquad \text{for } v_0 = (4k+2)\sqrt{\epsilon}
\end{equation*}
if $b_k = 1$ (see Figure \ref{fig:lower_bound}). First, we observe that all such functions have non-decreasing PDFs. Now, if two functions differ in a certain bit $b_t$ then the \levy-distance must be at least $\Omega(\epsilon)$ by comparing their CDF around $v_0 = (4k+2)\sqrt{\epsilon}$. In particular, let $b,b'\in\{0,1\}^n$ be two sequences that only differs at index $k$, with $b_k=0$ and $b'_k=1$. Let $F,G$ be the CDFs of distributions constructed by sequence $b$ and $b'$ respectively. Then $F$ and $G$ have $\Omega(\eps)$ \levy-distance since $F(v_0)<G(v_0-0.1\eps)-0.1\eps$.
\end{proof}

 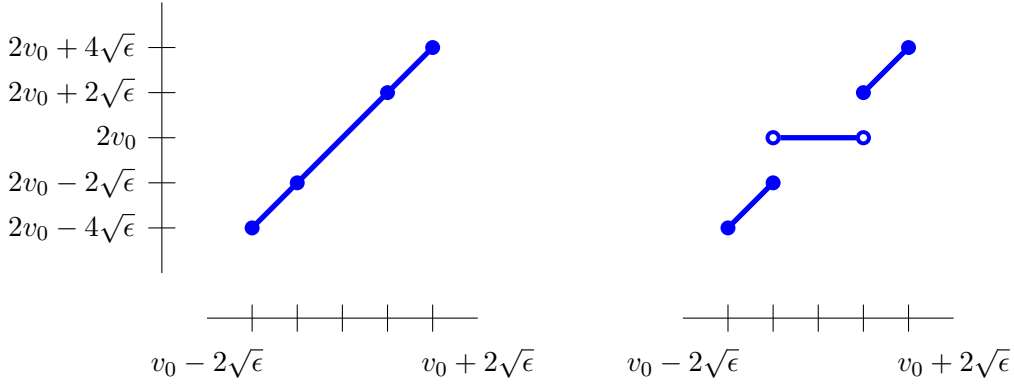
\begin{figure}[h]
\centering
\begin{tikzpicture}[scale=.6]

\node at (0,-1) {$v_0-2\sqrt{\eps}$};
\node at (6,-1) {$v_0 + 2\sqrt{\eps}$};

\draw (0,0) -- (6,0);
\draw (1,-.3) -- (1,.3);
\draw (2,-.3) -- (2,.3);
\draw (3,-.3) -- (3,.3);
\draw (4,-.3) -- (4,.3);
\draw (5,-.3) -- (5,.3);

\draw[line width=2, blue] (1,2)--(5,6);
\node[circle,fill,inner sep=2pt, blue] at (1,2) {};
\node[circle,fill,inner sep=2pt, blue] at (2,3) {};
\node[circle,fill,inner sep=2pt, blue] at (4,5) {};
\node[circle,fill,inner sep=2pt, blue] at (5,6) {};

\draw (-1,1)--(-1,7);
\draw (-1.3,3)--(-0.7,3);
\draw (-1.3,2)--(-0.7,2);
\draw (-1.3,4)--(-0.7,4);
\draw (-1.3,5)--(-0.7,5);
\draw (-1.3,6)--(-0.7,6);

\node at (-3,3) {$2v_0-2\sqrt{\eps}$};
\node at (-3,2) {$2v_0-4\sqrt{\epsilon}$};
\node at (-2,4) {$2v_0$};
\node at (-3,6) {$2v_0+4\sqrt{\epsilon}$};
\node at (-3,5) {$2v_0+2\sqrt{\eps}$};

\begin{scope}[xshift=300]
\node at (0,-1) {$v_0-2\sqrt{\eps}$};
\node at (6,-1) {$v_0 + 2\sqrt{\epsilon}$};

\draw (0,0) -- (6,0);
\draw (1,-.3) -- (1,.3);
\draw (2,-.3) -- (2,.3);
\draw (3,-.3) -- (3,.3);
\draw (4,-.3) -- (4,.3);
\draw (5,-.3) -- (5,.3);

\draw[line width=2, blue] (1,2)--(2,3);
\draw[line width=2, blue] (2,4)--(4,4);
\draw[line width=2, blue] (4,5)--(5,6);
\node[circle,fill,inner sep=2pt, blue] at (1,2) {};
\node[circle,fill,inner sep=2pt, blue] at (2,3) {};
\node[circle,fill,inner sep=2pt, white] at (2,4) {};
\node[circle, draw, inner sep=1.5pt, line width=1.5,blue] at (2,4) {};
\node[circle,fill,inner sep=2pt, white] at (4,4) {};
\node[circle, draw, inner sep=1.5pt, line width=1.5,blue] at (4,4) {};
\node[circle,fill,inner sep=2pt, blue] at (4,5) {};
\node[circle,fill,inner sep=2pt, blue] at (5,6) {};

\end{scope}

\end{tikzpicture}
\caption{The PDF $f_b$ in region $[v_0-2\sqrt{\eps}, v_0+2\sqrt{\epsilon}]$ for bit $b_k=0$ (left) and $b_k=1$ (right).}
\label{fig:lower_bound}
 \end{figure}

The above theorem shows that to describe a regular distribution within $\eps$ \levy-distance, $\Omega(\eps^{-0.5})$ bits are necessary. The main focus of the remaining section is to provide a matching upper bound $\tO(\eps^{-0.5})$ on the description complexity.

\subsection{Learning a smooth distribution}

We will start with the assumption that the distribution is smooth, by which we mean that it has no point masses and is described by a PDF of class $C^1$, i.e., the derivative $f'(v)$ exists and is continuous. Under this assumption, regularity can be written as the following condition:
\begin{equation}\label{eq:smooth_reg}
 f'(v) \cdot (1-F(v)) \geq- 2 f^2(v) 
\end{equation}
which corresponds to the derivative of the virtual value function being non-negative.  We start proving Lemma \ref{lem:levy-learn} with sufficient conditions for learning a smooth distribution. In Section \ref{sec:cvx} we describe a simplified version of our algorithm to learn a distribution with a convex PDF ($ f'(v) \geq 0$). After we describe our main ideas in this special case, we then turn to learn a smooth regular distribution in Section \ref{sec:reg}. Finally, we drop the smoothness assumption in Section \ref{sec:description-non-smooth} by using a regularity-preserving mollification argument.

\begin{lemma}\label{lem:levy-learn}
For any unknown distribution $F$ and $m$ points $x_1<x_2<\cdots<x_m$, such that for each $i$,
\begin{itemize}
\item we can identify $F_i\in[0,1]$ such that $F(x_i)\in [F_i-\eps,F_i+\eps]$, and
\item either $x_{i+1}-x_i\leq \eps$, or there exists $\underline{f}_i\leq \overline{f}_i$ such that $f(x)\in [\underline{f}_i,\overline{f}_i]$ for every $x\in[x_i,x_{i+1}]$, and satisfy $(x_{i+1}-x_i)(\overline{f}_i-\underline{f}_i)\leq \eps$,
\end{itemize}
we can construct a distribution $\hat{F}$ within $O(\eps)$ \levy-distance from $F$ on $[x_1,x_m]$.
\end{lemma}
\begin{proof}
We can assume that $F_1 \leq F_2 \leq \cdots \leq F_m$. If not, we can sort the values of $F_i$ in increasing order and the properties in the lemma will continue to hold. To see that, first  notice that if $F_i>F_{i+1}$ then we have $F(x_i)$ and $F(x_{i+1})$ are both in range $[F_{i}-\eps,F_{i+1}+\eps]$. This happens because  $F(x_i)\leq F(x_{i+1})$, and $F(x_i)\in [F_i-\eps,F_i+\eps]$, $F(x_{i+1})\in [F_{i+1}-\eps,F_{i+1}+\eps]$.
Then $F(x_i)\in[F_{i}-\eps,F_{i+1}+\eps]\subseteq [F_{i+1}-\eps,F_{i+1}+\eps]$, and $F(x_{i+1})\in[F_{i}-\eps,F_{i+1}+\eps]\subseteq [F_{i}-\eps,F_{i}+\eps]$. 
Therefore, if we switch $F_i$ and $F_{i+1}$, the properties in the lemma still hold.  

Consider the following distribution $\hat{F}$: For each $x_i$, $\hat{F}(x_i)=F_i$ and $$\hat{F}(x)=F_i+\frac{F_{i+1}-F_i}{x_{i+1}-x_i}(x-x_i) \text{ for } x\in(x_i,x_{i+1}).$$ In other words, $\hat{F}$ is defined by the estimation of $F$ at some points, and filling the curve in between by a linear function.

Now we show that $\hat{F}$ is within $O(\eps)$ \levy-distance from $F$ on $[x_1,x_m]$. We prove this by showing that this is true on every interval $[x_i,x_{i+1}]$.

If $x_{i+1}-x_{i}\leq\eps$, then $F$ and $\hat{F}$ are within $\eps$ \levy-distance on $[x_i,x_{i+1}]$ since for any $x \in [x_i,x_{i+1}]$
\[F(x-\eps)-\eps\leq F(x_i)-\eps\leq F_i\leq \hat{F}(x)\leq F_{i+1}\leq F(x_{i+1})+\eps\leq F(x+\eps)+\eps.\]

If $x_{i+1}-x_i>\eps$, we show a stronger statement that $\hat{F}$ is within $O(\eps)$ Kolmogorov distance from $F$ on $[x_i,x_{i+1}]$, i.e., for any $x\in[x_i,x_{i+1}]$, $|\hat{F}(x)-F(x)| \leq O(\eps)$. By the lemma statement, this is true for every $x_i$. Now we show that this is also true for every $x\in(x_i,x_{i+1})$.
Let $\hat{f}_i=\frac{F_{i+1}-F_i}{x_{i+1}-x_i}$ be the PDF of $\hat{F}$ on $(x_i,x_{i+1})$. Then
\begin{eqnarray}
|\hat{F}(x)-F(x)|&=&\left|(F_i+(x-x_i)\hat{f}_i)-(F(x_i)+\int_{x_{i}}^{x}f(t)dt)\right| \nonumber\\
&\leq&|F_i-F(x_i)|+\left|(x-x_i)\hat{f}_i-\int_{x_{i}}^{x}f(t)dt\right| \nonumber\\
&\leq&\eps+\int_{x_{i}}^{x}|\hat{f}_i-f(t)|dt\leq\eps+\int_{x_{i}}^{x_{i+1}}|\hat{f}_i-f(t)|dt.\label{eqn:dist-f-hatf}
\end{eqnarray}
Here the first line is by the definition of the CDF and PDF; the second line is by $|a+b|\leq |a|+|b|$ for any $a,b$; the third line is by the definition of $F_i$. Now we bound $\hat{f}_i$. As $f(x)\leq \overline{f}_i$ for every $x\in (x_i,x_{i+1})$, we have
\begin{equation*}
(x_{i+1}-x_i)\hat{f}_i-2\eps=(F_{i+1}-\eps) - (F_i+\eps)\leq F(x_{i+1})-F(x_i)=\int_{x_i}^{x_{i+1}}f(t)dt\leq (x_{i+1}-x_i)\overline{f}_i.
\end{equation*}
The same way, as $f(x)\geq \underline{f}_i$ for every $x\in (x_i,x_{i+1})$, we have
\begin{equation*}
(x_{i+1}-x_i)\hat{f}_i+2\eps=(F_{i+1}+\eps) - (F_i-\eps)\geq F(x_{i+1})-F(x_i)=\int_{x_i}^{x_{i+1}}f(t)dt\geq (x_{i+1}-x_i)\underline{f}_i.
\end{equation*}
Combine the two inequalities above, we can bound $\hat{f}_i$ by 
\[\underline{f}_i-\frac{2\eps}{x_{i+1}-x_i}\leq \hat{f}_i\leq\overline{f}_i+\frac{2\eps}{x_{i+1}-x_i}.\]
Thus for any $t\in (x_i,x_{i+1})$, as $f(t)\in [\underline{f}_i,\overline{f}_i]$, we have $|\hat{f}_i-f(t)|\leq \overline{f}_i-\underline{f}_i+\frac{2\eps}{x_{i+1}-x_i}$. Apply this bound to \eqref{eqn:dist-f-hatf}, we get
\begin{eqnarray*}
|\hat{F}(x)-F(x)|&\leq&\eps+\int_{x_{i}}^{x_{i+1}}|\hat{f}_i-f(t)|dt\\
&\leq&\eps+(x_{i+1}-x_i)\left(\overline{f}_i-\underline{f}_i+\frac{2\eps}{x_{i+1}-x_i}\right)\\
&=&3\eps+(x_{i+1}-x_i)(\overline{f}_i-\underline{f}_i)\leq 4\eps.
\end{eqnarray*}
\end{proof}

\subsection{Learning a smooth distribution with convex CDF}\label{sec:cvx}

In this section, we assume we have a smooth distribution $F$ and an oracle that allows us to sample the value of its CDF $F(\cdot)$ and its PDF $f(\cdot)$. We start with a special case of convex CDF ($f'(v) > 0$) to describe a simplified version of our algorithm. 

For a high-level intuition, it is useful to perform the following `heuristic' calculation based on Lemma \ref{lem:levy-learn}: if $\Delta x$ is the distance between query points and $\Delta f$ is the variation of the PDF between those points, our goal is to obtain $\Delta x \cdot \Delta f \leq \epsilon$. Approximating $\Delta f \approx f'(x) \cdot \Delta x$ we obtain $\Delta x \leq \sqrt{\epsilon / f'(x)}$. Hence, we will try to sample points at a rate $\sqrt{f'(x) / \epsilon}$.

We make this formal in the following proof and bound the number of queries we need to guarantee this sample rate.
Note that our proof is algorithmic and will be easily converted later to a pricing query complexity bound.

 \begin{theorem}
 \label{thm:convex}
 Let $F$ be a smooth distribution with a convex PDF. Then
$\tO(\eps^{-0.5})$ oracle queries to $F$ and $f$ are sufficient to learn the distribution within $\eps$  \levy-distance.
\end{theorem}

\begin{proof} We will show it is possible to find a set of $O(1/\epsilon^{0.5})$ points satisfying the conditions of Lemma \ref{lem:levy-learn}. We start by noticing that for a convex CDF, the PDF $f$ is monotone.



\paragraph{Step 1} We set $K=\log(1/\epsilon)$ and use binary search to identify points $x_0, x_1, \hdots, x_{2K+3}$ such that for each $k=0,\hdots, K+1$ we have $\abs{x_{2k+1} - x_{2k+2}} \leq \epsilon$ and for the interval $I_k = [x_{2k}, x_{2k+1}]$ it holds that:
\begin{itemize}
\item $f(v) \leq 1$ for $v \in I_0$
\item $2^{k-1} \leq f(v) \leq 2^k$ for $v \in I_k$, $k=1,...,K$
\item $f(v) \geq 2^K$ for $v \in I_{K+1}$
\end{itemize}
Identifying each point takes $O(\log(1/\epsilon)) = \tilde{O}(1)$ queries to $f(\cdot)$. Given that we have $\tilde{O}(1)$ such points, we used a total of $\tilde{O}(1)$ so far. We observe that the last interval $I_{K+1}$ already satisfies the conditions of Lemma \ref{lem:levy-learn} since its length $|I_{K+1}|$ is at most $\epsilon$:
\[1\geq \int_{v\in I_{K+1}}f(v)dv\geq 2^K |I_{K+1}|=\frac{1}{\eps}|I_{K+1}|.\]

\paragraph{Step 2} For each interval $I_k$ with $ k=0,..., K$, partition $I_k$ to intervals $I_{k,1},I_{k,2},\cdots$ of length $\frac{1}{2^k}\sqrt{\eps}$. The number of endpoints added is at most $\frac{|I_k|}{\sqrt{\eps}/2^k}=2^k\eps^{-1/2}|I_k|$ for each interval $I_k$, and sums up to at most
\begin{eqnarray*}
\sum_{k=0}^{K}2^k\eps^{-1/2}|I_k|&=&\eps^{-1/2}|I_0|+\sum_{k=1}^{K}2^k\eps^{-1/2}|I_k|\\
&\leq&\eps^{-1/2}+\sum_{k=1}^{K}\eps^{-1/2}\int_{v\in I_k}2f(v)dv\\
&\leq&\eps^{-1/2}+\eps^{-1/2}\int_{v\in [0,1]}2f(v)dv\\
&=&\eps^{-1/2}+2\eps^{-1/2}=3\eps^{-1/2}.
\end{eqnarray*}
Thus if we query $f$ and $F$ for all the endpoints of the intervals from the first two steps, the total number of queries needed is at most $O(\eps^{-1/2})$.

\paragraph{Step 3} For each interval $I_{k,j}=[v_{k,j},v_{k,j}+2^{-k}\sqrt{\eps}]\subseteq I_k$ partitioned in Step 2, define $\underline{f}_{k,j} = f(v_{k,j})$ and $\overline{f}_{k,j} =f(v_{k,j}+2^{-k}\sqrt{\eps}) $. If $\overline{f}_{k,j} - \underline{f}_{k,j} =t_{k,j}2^k\sqrt{\eps}$, then partition the interval to $\tilde{t}_{k,j}:=\max(1,\lceil{t_{k,j}}\rceil)$ intervals of length $\frac{1}{\tilde{t}_{k,j}2^k}\sqrt{\eps}$, and query $F$ for the newly added endpoints.
As any neighboring points $\overline{v}, \underline{v}$ have distance $\overline{v} - \underline{v} \leq \frac{1}{\tilde{t}_{k,j}2^k}\sqrt{\eps}$. Hence this interval satisfies the second condition in Lemma \ref{lem:levy-learn} :
$$ (\overline{v} -  \underline{v}) (\overline{f} -  \underline{f}) \leq \frac{1}{\tilde{t}_{k,j}2^k}\sqrt{\eps}\cdot t_{k,j}2^k\sqrt{\eps}\leq\eps$$
which allows us to learn the distribution up to $\epsilon$ \levy-distance.

Finally, we only need to bound the number of queries needed in this step. At most $\lfloor t_{k,j} \rfloor$ are needed in $I_{k,j}$ which has length $2^{-k} \sqrt{\epsilon}$. Now if $I_k = [x_{2k}, x_{2k+1}]$ we have:
\begin{eqnarray*}
2^{k-1} \geq f(x_{2k+1})-f(x_{2k})=\sum_{j}(f(v_{k,j+1})-f(v_{k,j}))=\sum_j t_{k,j}2^k\sqrt{\eps},
\end{eqnarray*}
which implies $\sum_{j}t_{k,j}\leq 1/(2\sqrt{\eps})$. Thus in this step, for all $k$, the total number of queries on $F$ is $\tO(\eps^{-1/2})$.
\end{proof}


\subsection{Learning a smooth regular distribution}\label{sec:reg}

Unlike convex functions, the PDF of a regular distribution can increase or decrease, so we can't easily partition the interval as in the previous lemma. However, regularity imposes a rate at which the PDF can decrease (equation \eqref{eq:smooth_reg}). Our first lemma shows that the PDF of a regular distribution can't decrease by a factor of $2$ too many times. We make this intuition formal in the next lemma:

\begin{lemma}\label{lem:regpdfdec}
Let $F$ be a smooth regular distribution with PDF $f$. Fix an integer $\ell$. Now assume $S$ is a collection of disjoint intervals of the form $[v_i, v'_i)$ such that:
\begin{itemize}
\item for each interval $[v_i, v'_i)$ we have $2^{-\ell}\leq 1-F(v'_i) \leq 1-F(v_i)<2^{-\ell+1}$
\item For each $v \in [v_i, v'_i)$ we have
$ f(v'_i) \leq f(v) \leq f(v_i)  $ and $f(v'_i) \leq \frac{1}{2}f(v_i)$
\end{itemize}
Then the number of intervals is bounded by a constant $\abs{S} \leq 32$.
\end{lemma}




\begin{proof}
We will count the maximum number of intervals in $S$. We can assume w.l.o.g. that $f(v'_i)=\frac{1}{2}f(v_i)$ as we can always take a subinterval where the last condition holds with equality.


Consider an interval $[v_i,v'_i)$ in the collection $S$ with $f(v_i)\in[2^{k-1},2^k)$. As for any $v\in[v_i,v'_i)$, 
\begin{equation}\label{eqn:df_lb}
f'(v)\geq-\frac{2f^2(v)}{1-F(v)}\geq -\frac{2f^2(v_i)}{1-F(v'_i)}\geq -\frac{2\cdot (2^k)^2}{2^{-\ell}}=-2^{2k+\ell+1}.
\end{equation}
Here the first inequality is by the definition of regular distributions; the second and the third inequalities are by $f$ is bounded between $f(v'_i)$, thus $f(v)\leq f(v_i)\leq 2^k$, and $1-F(v)\geq 1-F(v_i)\geq 2^{-\ell}$. Therefore 
\begin{equation}\label{eqn:interval-length}
-\frac{1}{2}\cdot 2^{k-1}\geq f(v'_i)-f(v_i)=\int_{v_i}^{v'_i}f'(v)dv\geq -2^{2k+\ell+1}(v'_i-v_i).
\end{equation}
Here the first inequality is by $f(v'_i)=\frac{1}{2}f(v_i)$ and $f(v_i)\geq 2^{k-1}$; the second inequality is by inequality \eqref{eqn:df_lb}. By inequality \eqref{eqn:interval-length} we can lower bound the length of the interval by $v'_i-v_i\geq\frac{1}{2}2^{k-1-(2k+\ell+1)}=2^{-k-\ell-3}$. Thus the CDF change in the interval can be lower bounded as follows:
\begin{equation}\label{eqn:interval-cdf-diff}
F(v'_i)-F(v_i)=\int_{v_i}^{v'_i}f(v)dv\geq f(v'_i)(v'_i-v_i)\geq 2^{k-2}\cdot2^{-k-\ell-3}=2^{-\ell-5}.
\end{equation}
Here the first inequality is by $f(v)\geq f(v_i)$ in the interval with a non-increasing PDF function, and the second inequality is by $f(v'_i)=\frac{1}{2}f(v_i)\geq 2^{k-2}$, and $v'_i-v_i\geq2^{-k-\ell-3}$. Since the intervals $[v_i,v'_i)$ are disjoint, the sum of $F(v'_i)-F(v_i)$ is at most $2^{-\ell+1}-2^{-\ell}=2^{-\ell}$ Therefore, there are at most $2^5=O(1)$ such intervals. 
\end{proof}


Lemma~\ref{lem:regpdfdec} shows that in an interval with the quantile bounded by a factor of 2, the PDF can decrease by a factor of $\frac{1}{2}$ for at most 32 times. A useful corollary is that in this case the PDF can decrease by a constant factor of at most $1-2^{-33}$ between two points.

\begin{corollary}\label{cor:pdfdec-32}
Let $F$ be a smooth regular distribution and $\ell$ an integer. 
For any two points $x<x'$ with $2^{-\ell}\leq 1-F(x')<1-F(x)\leq 2^{-\ell+1}$ it holds that $f(x')>2^{-33}f(x)$.
\end{corollary}

\begin{proof}
If $f(x)\geq 2^{33}f(x')$, as $f$ is continuous, $[x,x')$ can be partitioned to 33 intervals $[x,v_1)=[v_0,v_1)$, $[v_1,v_2)$, $[v_2,v_3)$, $\cdots$, $[v_{31},v_{32})$, $[v_{32},x')=[v_{32},v_{33})$ such that for any interval $[v_i,v_{i+1})$, $f(v_{i+1})\leq \frac{1}{2}f(v_i)$ which contradicts Lemma~\ref{lem:regpdfdec}.
\end{proof}

With Lemma~\ref{lem:regpdfdec} we can also obtain the following key component of the learning algorithm.

\begin{lemma}\label{lem:binary-search-query}
Let $F$ be a smooth regular distribution and $\ell$ an integer and $\epsilon > 0$ a parameter.
For any interval $[\ux,\ox]$ with $2^{-\ell}\leq 1-F(\ox)<1-F(\ux)\leq 2^{-\ell+1}$, using $\tO(1)$ queries to $f$ we can find a sequence of points $\ux=x_1\leq x_2\leq\cdots\leq x_m=\ox$ in the interval, such that: for any sub-interval $[x_i,x_{i+1})$, either
\begin{itemize}
\item $x_{i+1}-x_i\leq \eps$;
\item $f(x_{i+1})<2^{-33}$, which implies for any point $x\in[x_i,x_{i+1})$, $f(x)\leq 1$;
\item for any two points $x,x'\in [x_i,x_{i+1})$, $2^{-36}f(x)<f(x')<2^{36}f(x)$.
\end{itemize}
\end{lemma}

\begin{proof}
Consider the following binary search algorithm: For any interval $[a,b]$ (begin with $[a,b]=[\ux,\ox-\eps]$ and only $\ux,\ox-\eps,\ox$ in the output sequence),
\begin{enumerate}
\item If $b-a\leq \eps$, stop searching $[a,b]$;
\item If $f(b)<4f(a)$, stop searching $[a,b]$;
\item If $f(b)<2^{-33}$, stop searching $[a,b]$;
\item Otherwise, add $\frac{a+b}{2}$ to the output sequence, then search both $[a,\frac{a+b}{2}]$ and $[\frac{a+b}{2},b]$.
\end{enumerate}
We show that all the first three cases in the algorithm are valid stopping conditions.

In Case 1, we don't need to search the interval $[a,b]$ as the endpoints are close enough to satisfy the required property of the lemma. 

In Case 2, consider interval $[a,b]$ with $f(b)\leq 4f(a)$, and any $x<x'$ in $[a,b]$. By Corollary~\ref{cor:pdfdec-32} $f(x')<2^{36}f(x)$.
If $f(x')\geq 2^{36}f(x)$, as $f(b)<4f(a)$, there exist at least 33 intervals $[a,v_1)$, $[v_1,v_2)$, $\cdots$, $[v_{j-2},v_{j-1})$, $[v_{j-1},x)$, $[x',v_{j+1})$, $[v_{j+1},v_{j+2})$, $\cdots$, $[v_{31},v_{32})$, $[v_{32},b)$ such that for any interval listed above, the PDF between the two endpoints decreases by at least a factor of $\frac{1}{2}$. This also contradicts Lemma~\ref{lem:regpdfdec} that in $[\ux,\ox]$ the PDF can only decrease by half for at most 32 times. Similarly, $f(x') > 2^{-36} f(x)$ otherwise there are at least $33$ intervals between $x$ and $x'$ where the PDF decreases by half.

In Case 3, if $f(b)<2^{-33}$, by Corollary~\ref{cor:pdfdec-32} there does not exist any $x\in[a,b)$ such that $f(x)>1>2^{33}f(b)$. Thus for any $x\in[a,b)$, $f(x)\leq 1$.\\

Thus, all first 3 cases of the algorithm are valid stopping conditions. Now we analyze how many points are added to the sequence. First notice that for any $v\leq\ox-\eps$, $f(v)<\frac{2^{33}}{\eps}$: otherwise if $f(v)\geq\frac{2^{33}}{\eps}$, as $f(x)>2^{-33} f(v)\geq\frac{1}{\eps}$ for every $x\in [v,v+\eps]$ by Corollary~\ref{cor:pdfdec-32}, we have $F(v+\eps)-F(v)=\int_{v}^{v+\eps}f(t)dt>\frac{1}{\eps}((v+\eps)-v)=1$, which is impossible.

Consider the following set $S$ containing all searched intervals $[a_i,b_i]$ that are one level from the leaves: in other words, the algorithm searches $[a_i,\frac{a_i+b_i}{2}]$ and $[\frac{a_i+b_i}{2},b]$, but does not search deeper into either interval. If we list all intervals in $S$ in increasing order (of either endpoint) $[a_1,b_1], [a_2,b_2], \cdots, [a_s,b_s]$, then $f$ increases by a factor of 4 from $a_i$ to $b_i$, but decreases by at least half from $b_i$ to $a_{i+1}$ for at most 32 times due to Lemma~\ref{lem:regpdfdec}. Therefore, $f(b_{i+1})>2f(b_i)$ for all but at most 32 times. Partition $S$ to at most 32 subsequence of intervals $S_d=([a_{d+1},b_{d+1}]$, $[a_{d+2},b_{d+2}]$, $\cdots$, $[a_{d+j},b_{d+j}]$) with $f(b_{d+i+1})>2f(b_{d+i})$ for every $i\in[j-1]$. Since for every $[a_i,b_i]\in S$, $2^{-33}\leq f(b_i)<\frac{2^{33}}{\eps}$, $S_d$ contains at most $O(\log\frac{1}{\eps})=\tO(1)$ intervals, thus $S$ also contains $\tO(1)$ intervals.

Notice that each interval in the search tree satisfies one of the 4 conditions:
\begin{enumerate}
    \item The interval is in $S$;
    \item The interval is a child of an interval in $S$;
    \item In the search tree, the interval is an ancestor of an interval in $S$;
    \item In the search tree, the interval is a child of an ancestor of an interval in $S$.
\end{enumerate}
Since the length of each searched interval is at least $\eps$, the depth of the search tree is $O(\log\frac{1}{\eps})=\tO(1)$. Therefore, each interval $[a_i,b_i]\in S$ defined in the previous paragraph has $\tO(1)$ ancestors in the search tree. Thus each interval in $S$ can define at most $\tO(1)$ intervals in the above 4 conditions, which means there are $\tO(|S|)=\tO(1)$ intervals in the search tree. We conclude that the proposed algorithm uses $\tO(1)$ queries to $f$ to output a sequence of points satisfying the conditions in the lemma statement. 

\end{proof}

Now we are ready to generalize the algorithm for distributions with monotone PDF to arbitrary regular distributions. The algorithm has almost the same structure, but we need to be more careful in the analysis of query complexity.

\begin{theorem}
\label{thm:regular}
Let $F$ be a smooth regular distribution. Then
$\tO(\eps^{-0.5})$ oracle queries to $F$ and $f$ are sufficient to learn the distribution within $\eps$  \levy-distance.

\end{theorem}
\begin{proof}

Let $L=\log\frac{1}{\eps}$. We describe the steps of the algorithm as follows. We start with ``Step 0'' just to keep the correspondence of Steps 1, 2, and 3.

\paragraph{Step 0} Use binary search to identify $2(L+1)$ points $z_i$ such that for the intervals $I^{(\ell)} = [z_{2\ell-1}, z_{2\ell}]$ the CDF satisfies $1-2^{-\ell+1}\leq F(v)<1- 2^{-\ell}$ and $I^{(L+1)}$ contains all values with $1-F(v)\geq 1-2^{-L}=1-\eps$. And $z_{2\ell+1} - z_{2\ell} \leq \epsilon$. Binary search requires 
$O(\log\frac{1}{\eps})=\tilde{O}(1)$ queries for each interval, thus $\tilde{O}(1)$ queries in total for all intervals as there are only $\tilde{O}(1)$ intervals. We are going to run an algorithm similar to the case with monotone PDF for each interval and show that $\tilde{O}(\eps^{-0.5})$ queries to $F$ and $f$ are enough to learn the CDF for each $I^{(\ell)}$, thus $\tilde{O}(\eps^{-0.5})$ queries are enough in total.

\paragraph{Step 1} All operations below are applied to interval $I^{(\ell)}$ with fixed $\ell$. We only focus on $I^{(\ell)}$ with $\ell\leq L$, as there is no need to learn $I^{(L+1)}$ with $\eps$ Levy distance. When there is no ambiguity, we omit the superscript $\ell$. 

For interval $I^{(\ell)}=[\ux,\ox]$, by Lemma~\ref{lem:binary-search-query} we can find a sequence of points $x_1=\ux\leq x_2\leq\cdots\leq x_m=\ox$ in the interval, such that: for any subinterval $[x_i,x_{i+1})$, either (a). $x_{i+1}-x_i\leq \eps$; or (b). $f(x_{i+1})<2^{-33}$, which implies for any point $x\in[x_i,x_{i+1})$, $f(x)\leq 1$; or (c). for any two points $x,x'\in [x_i,x_{i+1})$, $2^{-36}f(x)<f(x')<2^{36}f(x)$. 

Construct interval collections $I_0,I_1,\cdots,I_{K},I_{K+1}$ for some $K=\tO(1)$ as follows.
For all subintervals $[x_i,x_{i+1})$ satisfying (a), group them to an interval collection $I_{K+1}$; for all subintervals satisfying (b), group them to an interval collection $I_{0}$; for all subintervals satisfying (c), add $[x_i,x_{i+1})$ to interval collection $I_k$, if $2^{k-1}\leq 2^{33}f(x_{i+1})<2^{k}$. Different from the analysis for monotone non-decreasing $f$, for any interval $[x_i,x_{i+1})$ in $I_k$, point $x\in[x_i,x_{i+1})$ satisfies $f(x)<2^{36}f(x_{i+1})<2^{k+3}$, and $f(x)>2^{-36}f(x_{i+1})\geq 2^{k-70}$. As the intervals in the collections are constructed by Lemma~\ref{lem:binary-search-query} which only uses $\tO(1)$ queries, the total number of queries to $f$ used in this step is $\tO(1)$.

\paragraph{Step 2} For each interval collection $I_k$ with $0\leq k\leq K$, partition all points in $I_k$ to intervals $I_{k,1},I_{k,2},\cdots$ of length $\frac{1}{2^{k+\ell}}\sqrt{\eps}$. Denote by $|I_k|$ the total length of all intervals in $I_k$. The number of endpoints added is at most $\frac{|I_k|}{\sqrt{\eps}/2^{k+\ell}}=2^{k+\ell}\eps^{-1/2}|I_k|$ for each interval collection $I_k$, and sums up to at most 
\begin{eqnarray*}
\sum_{k=0}^{K}2^k\eps^{-1/2}|I_k|&=&\eps^{-1/2}|I_0|+\sum_{k=1}^{K}2^{k+\ell}\eps^{-1/2}|I_k|\\
&\leq&\eps^{-1/2}+\sum_{k=1}^{K}\eps^{-1/2}\int_{v\in I_k}2^{70+\ell}f(v)dv\\
&\leq&\eps^{-1/2}+2^{70+\ell}\eps^{-1/2}\int_{v\in I^{(\ell)}}f(v)dv\\
&=&\eps^{-1/2}+2^{70+\ell}\eps^{-1/2}\cdot2^{-\ell}=O(\eps^{-1/2}).
\end{eqnarray*}
Thus if we query $f$ and $F$ for all the endpoints of the intervals from this step, the total number of queries needed is at most $O(\eps^{-1/2})$.

\paragraph{Step 3} For each interval $I_{k,j}=[v_{k,j},v_{k,j}+2^{-k-\ell}\sqrt{\eps}]\subseteq I_k$ partitioned in Step 2, if $f(v_{k,j}+2^{-k-\ell}\sqrt{\eps})-f(v_{k,j})=t_{k,j}\cdot 2^{k+\ell}\sqrt{\eps}$, we further partition the interval to $\tilde{t}_{k,j}:=\max(1,\lceil{t_{k,j}}\rceil)$ intervals $[v_{k,j,0},v_{k,j,1}]$, $[v_{k,j,1},v_{k,j,2}]$, $\cdots$, $[v_{k,j,\tilde{t}_{k,j}-1},v_{k,j,\tilde{t}_{k,j}}]$ of length $\frac{1}{\tilde{t}_{k,j}2^{k+\ell}}\sqrt{\eps}$ (with $v_{k,j,0}=v_{k,j}$ and $v_{k,j,\tilde{t}_{k,j}}=v_{k,j+1}$), and query $F$ for the newly added endpoints. By the definition of regular distribution, for any $x\in [v_{k,j},v_{k,j+1}]$,
\[f'(x)\geq -\frac{2f^2(x)}{1-F(x)}>\frac{2\cdot (2^{k+3})^2}{2^{-\ell}}=-2^{\ell+2k+7}\]
by $f(x)<2^{k+3}$ and $F(x)<1-2^{\ell}$. Thus we can bound $f(x)$ for $x\in [v_{k,j},v_{k,j+1}]$ as follows: 
\begin{eqnarray*}
f(x)&\leq& f(v_{k,j+1})-(-2^{\ell+2k+7})\cdot(v_{k,j+1}-v_{k,j})\\
&<&f(v_{k,j})+t_{k,j}\cdot 2^{k+\ell}\sqrt{\eps}+2^{\ell+2k+7}2^{-k-\ell}\sqrt{\eps}\\
&=&f(v_{k,j})+t_{k,j}\cdot2^{k+\ell}\sqrt{\eps}+2^{k+7}\sqrt{\eps}\mathop{=}^{def} \of,
\end{eqnarray*}
and 
\begin{eqnarray*}
f(x)&\geq& f(v_{k,j})+(-2^{\ell+2k+7})\cdot(v_{k,j+1}-v_{k,j})\\
&>&f(v_{k,j})-2^{\ell+2k+7}2^{-k-\ell}\sqrt{\eps}\\
&=&f(v_{k,j})-2^{k+7}\sqrt{\eps}\mathop{=}^{def} \uf,
\end{eqnarray*}
Thus for any neighboring points $v_{k,j,a},v_{k,j,a+1}$ in the partition with distance $d=\frac{1}{\tilde{t}_{k,j}2^{k+\ell}}\sqrt{\eps}$, if $x\in[v_{k,j,a},v_{k,j,a+1}]$, we have $f(x)\in(\uf,\of)$ with $\of-\uf<t_{k,j}\cdot2^{k+\ell}\sqrt{\eps}+2^{k+8}\sqrt{\eps}$, and
\begin{eqnarray*}
(v_{k,j,a+1}-v_{k,j,a})(\of-\uf)&<&\frac{1}{\tilde{t}_{k,j}2^{k+\ell}}\sqrt{\eps}(t_{k,j}\cdot2^{k+\ell}\sqrt{\eps}+2^{k+8}\sqrt{\eps})\\
&=&\frac{t_{k,j}}{\tilde{t}_{k,j}}\eps+2^{8-\ell}\eps=O(\eps).
\end{eqnarray*}


\begin{figure}[h]
\centering
\includegraphics[width=0.9\textwidth]{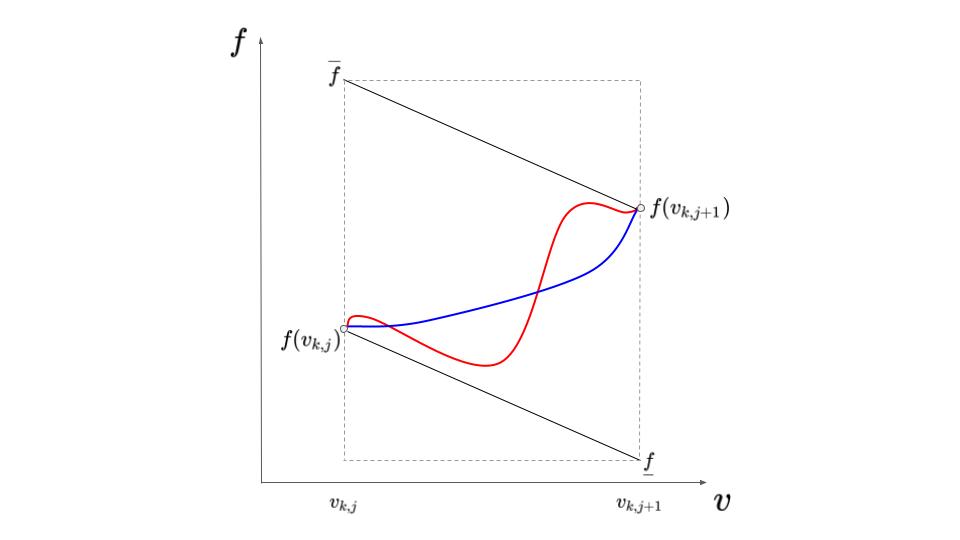}
\caption{Illustration of the analysis on interval $[v_{k,j},v_{k,j+1}]$ with $t_{k,j}\leq 1$. The two black solid lines have slope $-2^{\ell+2k+7}$, which is a lower bound of $f'(v)$ in this range by the definition of regular distribution. Given $f(v_{k,j})$ and $f(v_{k,j+1})$, the PDF function $f$ of a regular distribution cannot go beyond the two black lines. As the CDF of any distribution is the area below the PDF curve, the CDF difference of any two regular distributions given the CDF and the PDF at the $v_{k,j}$ and $v_{k,j+1}$ (for example, the red curve and the blue curve) is smaller than the area of the dashed rectangle, which is $(v_{k,j+1}-v_{k,j})(\of-\uf)$.}
\label{fig:regular-example}
\end{figure}

Figure~\ref{fig:regular-example} provides an illustrative example of the analysis. By Lemma~\ref{lem:levy-learn} we can learn a distribution within $\eps$ \levy-distance from $F$ on $\bigcup I_{k,j}$ for $k\in[0,K]$. As all intervals in $I_{K+1,j}$ have length at most $\eps$, we do not need to learn those intervals to satisfy $\eps$ \levy-distance. 

Now we analyze how many additional points (i.e. queries to $F$) are needed in this step. $\tilde{t}_{k,j}-1=\max(1,\lceil t_{k,j}\rceil)-1\leq \max(0,t_{k,j})$ additional queries are needed for each interval $I_{k,j}$. This means that if $f$ increases on $I_{i,j}$, i.e. $f(v_{k,j+1})>f(v_{k,j})$, then at most $t_{k,j}=\frac{f(v_{k,j+1})-f(v_{k,j})}{2^{k+\ell}\sqrt{\eps}}$ additional queries are needed; otherwise when $f$ does not increase on $I_{i,j}$, no additional query is required. 

Let $I_{k,\uparrow}=\{x:f'(x)\geq 0,2^{k-70}< f(x)< 2^{k+3}\}$ and $I_{k,\downarrow}=\{x:f'(x)<0,2^{k-70}< f(x)< 2^{k+3}\}$ be the sets of values such that $f(x)$ are bounded in range $(2^{k-70},2^{k+3})$, while $f$ is non-decreasing and decreasing respectively. Notice that $I_{k}$ is a subset of $I_{k,\uparrow}\cup I_{k,\downarrow}$. Therefore, we can upper-bound the total number of additional queries in $I_k$ by 
\begin{equation}\label{eqn:additional-num}
\sum_{j}\max(0,t_{k,j})\leq \sum_{j} \frac{f(v_{k,j+1})-f(v_{k,j})}{2^{k+\ell}\sqrt{\eps}}
\leq\frac{1}{2^{k+\ell}\sqrt{\eps}}\int_{I_{k,\uparrow}}f'(x)dx.
\end{equation}

As on $I_{k,\uparrow}\cup I_{k,\downarrow}$, $f(x)$ can increase to at most $2^{k+3}$, we have $\int_{I_{k,\uparrow}}f'(x)dx\leq 2^{k+3}-\int_{I_{k,\downarrow}}f'(x)dx$. Notice that for any $x$ such that $2^{k-70}< f(x)< 2^{k+3}$, when $f'(x)<0$, $f'(x)>-\frac{2f^2(x)}{1-F(x)}>-2^{2k+\ell+7}$. Also by 
\[2^{k-70}|I_{k,\downarrow}|<\int_{I_{k,\downarrow}}f(x)dx\leq 2^{-\ell}\] 
we have $|I_{k,\downarrow}|<2^{70-k-\ell}$.
Therefore,
\begin{eqnarray}
\int_{I_{k,\uparrow}}f'(x)dx&\leq& 2^{k+3}-\int_{I_{k,\downarrow}}f'(x)dx<2^{k+3}+\int_{I_{k,\downarrow}}2^{2k+\ell+7}dx\nonumber\\
&<&2^{k+3}+2^{2k+\ell+7}|I_{k,\downarrow}|<2^{k+3}+2^{2k+\ell+7}\cdot 2^{70-k-\ell}\nonumber\\
&<&2^{k+3}+2^{k+77}<2^{k+78}.\label{eqn:uprange}
\end{eqnarray}
Combine \eqref{eqn:additional-num} and \eqref{eqn:uprange}, we have the total number of additional query in $I_k$ is at most 
\[\frac{1}{2^{k+\ell}\sqrt{\eps}}\int_{I_{k,\uparrow}}f'(x)<\frac{1}{2^{k+\ell}\sqrt{\eps}}\cdot 2^{k+78}=2^{78-\ell}\sqrt{\eps}.\]
Thus for every $k$, the number of queries in Step 3 is $O(\eps^{-1/2})$, which means that the total number of queries in Step 3 is $\tO(\eps^{-1/2})$ since there are $\tO(1)$ different values of $k$.

\end{proof}


\paragraph{Discussion on $\beta$-irregular distributions.} Observe that there is nothing special about the factor of $2$ in equation \eqref{eq:smooth_reg}. By repeating the proof with a different irregularity coefficient we obtain a proof of Theorem \ref{thm:beta-irregular}.





\subsection{Description complexity of a general (non-smooth) regular distribution}
\label{sec:description-non-smooth}

When the regular distribution does not have a smooth CDF function, the PDF function may not exist, and our algorithm for learning the CDF distribution using oracle queries to the PDF function may not apply to general regular distributions. However, we show in the following lemma, that any regular distribution $F$ can be uniformly approximated by a smooth regular distribution $\tilde{F}$ within arbitrarily small \levy-distance $\delta$. Thus, to describe $F$ within $\eps$ \levy-distance, it suffices to describe $\tilde{F}$ within $\eps$ \levy-distance, which is doable via the algorithm in previous subsections.

What we will describe above is a variation of the well-known mollification procedure in real analysis, which replaces a function with its convolution with a very concentrated $C^\infty$-function. The only specific detail below is that we argue we can do mollification in a way that is regularity-preserving.

\begin{lemma}\label{lem:smooth}
Given any regular distribution $F$, for any $\delta>0$ there is another regular distribution $\tilde{F}$ such that $\text{\levy}(F, \tilde{F}) \leq \delta$ and $\tilde{F}$ has no point-masses and has a $C^\infty$ CDF.
\end{lemma}

\begin{proof}
Let $R(q) = q F^{-1}(1-q)$ be the revenue curve associated $F$. Since $F$ is regular, the curve $R$ is concave, but it may not be $C^\infty$. There are different techniques to obtain a uniform approximation of concave functions by $C^\infty$ concave functions, one of which is taking a convolution with a non-negative mollifier (see \cite{azagra2015global} for example). For our purpose, it is enough that it exists a function $C^\infty$-concave $\hat{R}(x)$ such that $\hat{R}(0) = 0$ and $\abs{R(x) - \hat{R}(x)} \leq \delta^2/2, \forall x \in [0,1]$. Now, if we take $\tilde{R}(x) = \hat{R}(x) + \delta^2 x(1-x)$ we obtain a function that is $C^\infty$ strongly concave and $\abs{R(x) - \tilde{R}(x)} \leq \delta^2$.

Now, define $\tilde{F}$ such that $\tilde{R}(q) = q \tilde{F}^{-1}(1-q)$. Because $\tilde{R}$ is strongly convex, $\tilde{F}$ has no point masses. The function $\tilde{F}$ is strictly increasing and $C^\infty$, so $\tilde{f}(v) > 0$ in its domain. As a consequence, $\tilde{F}(v)$ is also $C^\infty$.

Finally observe that for $q \in [\delta, 1-\delta]$ we have $\abs{\tilde{F}^{-1}(q) - F^{-1}(q)} \leq \delta^2 / \delta = \delta$. Now, for a given $q$, let $v = F^{-1}(q)$. Then we know that:
$v-\delta \leq \tilde{F}^{-1}(q) \leq v+\delta$. Hence
$\tilde{F}(v-\delta) \leq q = F(v) \leq \tilde{F}(v+\delta)$
which implies that $\text{\levy}(F,\tilde{F}) \leq \delta$.
\end{proof}


\paragraph{Putting it all together} 
We can apply Theorem \ref{thm:regular} to the smooth approximation from Lemma \ref{lem:smooth} and obtain a $\tilde{O}(1/\epsilon^{0.5})$ bound on the number of queries to $F$ and $f$ to learn any regular distribution.

To complete the result, we still need to argue about the number of bits necessary to represent the result since both the points $x$ queried and the results $F(x)$ and $f(x)$ are in principle real numbers. However, we can see each query to $F$ (or $f$) as follows: query $F(x)$ (or $f(x)$)  where $x$ is an integer times $\eps$, and an approximate answer within $F(x)\pm\epsilon$ (or $f(x)\pm\epsilon$). All of the lemmas and algorithms already accommodate $\pm \epsilon$ errors both in $x$, $F$, and $f$. This automatically leads to a $\tilde{O}(1/\epsilon^{0.5})$ bit complexity bound. Combined with the lower bound in Theorem~\ref{thm:description-lb}, this proves our main result (Theorem~\ref{thm:main}).

\section{Pricing Query Complexity of Learning a Regular Distribution}\label{sec:mol}

In this section, we study the pricing query complexity of learning a regular distribution within $\eps$ \levy-distance. In particular, for each pricing query, we submit a price $x$, and a value $v\sim F$ is realized. We cannot directly observe $v$, but can observe $\text{sign}(v-x) \in \{-1,+1\}$. By Chernoff bound, $\tilde{O}(1/\eps^2)$ pricing queries are sufficient to learn $F(x)$ with $\eps$ additive error. However, the $\tilde{O}(\eps^{-0.5})$ description complexity upper bound result is not sufficient to show that the pricing query complexity of learning a regular distribution on $[0,1]$ within $\eps$ \levy-distance. There are two problems remaining. Firstly, the algorithms for showing description complexity upperbound results rely on queries to not only the CDF, but also the PDF. Secondly, for non-smooth regular distributions, we prove the description complexity result by learning a smooth regular distribution $\tF$ that is uniformly close to $F$, but not exactly $F$. Thus when we query $F(x)$, we are only able to get $\tF(x+\eps_1)+\eps_2$ where $|\eps_1|$ can be arbitrarily small (and we use the following bound $|\eps_1|\leq \frac{\eps^2}{2}$), and $|\delta_2| < \epsilon$. Thus, the problem we want to solve is equivalent to the following theorem: 

\begin{theorem}
\label{thm:perturbed-regular}
For any regular distribution $F$ with smooth PDF function, suppose that we have perturbed oracle access to $F$ that every time when querying $x$, the oracle $F^*$ returns $F^{*}(x)=F(x+\eps_1)+\eps_2$ for some unknown $\eps_1\in[-\frac{\eps^2}{2},\frac{\eps^2}{2}]$, and $\eps_2\in[-\eps,\eps]$. Then $\tilde{O}(\eps^{-0.5})$ queries to the perturbed oracle $\Fs$ are sufficient to learn $F$ within $\eps$ \levy-distance.
\end{theorem}


\subsection{Query complexity of distributions with convex $F$}
Before proving Theorem~\ref{thm:perturbed-regular} for general regular distributions, we first prove the theorem for a distribution with convex $F$ to give more intuition. The spirit of the general case is almost identical.

\begin{theorem}
\label{thm:perturbed-convex}
For any regular distribution $F$ with smooth PDF function, suppose that we have perturbed oracle access to $F$ that every time when querying $x$, the oracle $F^*$ returns $F^{*}(x)=F(x+\eps_1)+\eps_2$ for some unknown $\eps_1\in[-\frac{\eps^2}{2},\frac{\eps^2}{2}]$, and $\eps_2\in[-\eps,\eps]$. Then $\tilde{O}(\eps^{-0.5})$ queries to the perturbed oracle $\Fs$ are sufficient to learn $F$ within $\eps$ \levy-distance.
\end{theorem}

\begin{proof}
We first show how to define a (perturbed) query to $f$ via the perturbed oracle. Suppose that we want to query $f$. When we query $x+\gamma$ and $x-\gamma$ to the perturbed oracle $\Fs$, we will get $\Fs(x+\gamma)=F(x+\gamma+\eps_1)+\eps_2$, and $\Fs(x-\gamma)=F(x-\gamma+\eps_3)+\eps_4$ for some unknown $\eps_1,\eps_3\in[-\frac{\eps^2}{2},\frac{\eps^2}{2}]$, and $\eps_2,\eps_4\in[-\eps,\eps]$. Let 
\[\fs_\gamma(x)=\frac{\Fs(x+\gamma)-\Fs(x-\gamma)}{2\gamma}.\]
By Lagrange mean value theorem, for smooth function $F$, there exists $x^*\in [x-\gamma+\eps_3,x+\gamma+\eps_1]$ such that 
\begin{eqnarray*}
f(x^*)&=&\frac{F(x+\gamma+\eps_1)-F(x-\gamma+\eps_3)}{(x+\gamma+\eps_1)-(x-\gamma+\eps_3)}\\
&=&\frac{(\Fs(x+\gamma)-\eps_2) - (\Fs(x-\gamma) - \eps_4)}{2\gamma+\eps_1-\eps_3}\\
&=&\frac{2\gamma\fs_\gamma(x)-\eps_2+\eps_4}{2\gamma+\eps_1-\eps_3}.
\end{eqnarray*}

Then we have
\begin{equation}\label{eqn:pdf-query}
\fs_{\gamma}(x)=\frac{2\gamma+\eps_1-\eps_3}{2\gamma}f(x^*)+\frac{\eps_2-\eps_4}{2\gamma}=(1+\delta_1)f(x+\gamma_1)+\delta_2,
\end{equation}
for some $\delta_1,\delta_2,\gamma_1$ where $|\delta_1|\leq\frac{\eps^2}{2\gamma}$, $|\delta_2|\leq\frac{\eps}{\gamma}$, $|\gamma_1|<2\gamma$. Later, we will make perturbed queries to $f$ via $\fs_{\gamma}(x)$, and specify the accuracy parameter $\gamma$ needed in each step of the learning algorithm.

\paragraph{\textit{Changes to Lemma~\ref{lem:levy-learn}.}} Suppose that we are only able to identify $F_i\in[0,1]$ such that $F(x_i+\eps_i)\in [F_i-\eps,F_i+\eps]$ for some unknown small $\eps_i\in[-\eps^2/2,\eps^2/2]$. We show that Lemma~\ref{lem:levy-learn} still holds. Observe that the distribution constructed from input $(x_1,\cdots,x_m)$ and $(F_1,\cdots,F_m)$ in the proof of Lemma~\ref{lem:levy-learn} is only $O(\eps)$ \levy-distance from the distribution constructed from input $(x_1',x_2',\cdots,x_m')=(x_1+\eps_1,x_2+\eps_2,\cdots,x_m+\eps_m)$ and $(F_1,\cdots,F_m)$. Furthermore, as either $x'_{i+1}-x'_i$ is $O(\eps)$, or $(x'_{i+1}-x'_i)(\of_i-\uf_i)<2(x_{i+1}-x_i)(\of_i-\uf_i)=O(\eps)$ (since $x'_{i+1}-x'_i\leq x_{i+1}-x_{i}+\eps^2 < 2(x_{i+1}-x_{i})$), the distribution constructed from input $(x_1',\cdots,x_m')$ and $(F_1,\cdots,F_m)$ is within $(\eps)$ \levy-distance from $F$. Thus the distribution constructed from input $(x_1,\cdots,x_m)$ and $(F_1,\cdots,F_m)$ is within $O(\eps)$ \levy-distance from $F$ on $[x_1,x_m]$. To summarize, to apply Lemma~\ref{lem:levy-learn}, we do not need $F_i$ to be an estimate of $F(x)$ with $O(\eps)$ error; letting $F_i=F^*(x)$ is good enough.

\paragraph{\textit{Changes to Theorem~\ref{thm:convex}}.} We now describe in detail when $F$ is convex (i.e. with monotone $f$), how does the algorithm change.

\paragraph{Step 1} We set $K=\log(1/\epsilon)$ and use binary search to identify points $x_0, x_1, \hdots, x_{2K+3}$ such that for each $k=0,\hdots, K+1$ we have $\abs{x_{2k+1} - x_{2k+2}} =O(\epsilon)$ and for the interval $I_k = [x_{2k}, x_{2k+1}]$ it holds that: 
\begin{itemize}
\item $f(v) \leq 2$ for $v \in I_0$
\item $2^{k-2} \leq f(v) \leq 2^{k+1}$ for $v \in I_k$, $k=1,...,K$
\item $f(v) \geq 2^{K-1}$ for $v \in I_{K+1}$
\end{itemize}
Notice that now we do not have access to $f$. Thus in order to perform the binary search, we query $\fs_{\gamma}$ with $\gamma=4\eps$. By \eqref{eqn:pdf-query}, for any query $x$, $\fs_{\gamma}(x)=(1+\delta_1)f(x+\gamma_1)+\delta_2$ for some $\delta_1,\delta_2,\gamma_1$ with $|\delta_1|\leq\frac{\eps^2}{2\gamma}\leq\frac{1}{8}$, $|\delta_2|\leq \frac{\eps}{\gamma}=\frac{1}{4}$, $|\gamma_1|\leq2\gamma=8\eps$. Thus
\[\frac{7}{8}f(x-8\eps)-\frac{1}{4}\leq \fs_{\gamma}(x)\leq\frac{9}{8}f(x+8\eps)+\frac{1}{4}.\]
Therefore, when $\fs_{\gamma}(x)\geq 2^{k-1}$, $f(x+8\eps)>2^{k-2}$; when $\fs_{\gamma}(x)\leq 2^{k}$, $f(x-8\eps)<2^{k+1}$. Thus, suppose we do binary search on $\fs_{\gamma}$ to find a sequence of points $x'_0,x'_1,\cdots,x'_{2K+3}$ such that for every $k=0,\cdots,K$, $\fs_{\gamma}(x'_{2k+1})\leq 2^{k}$, $\fs_{\gamma}(x'_{2k+2})\geq 2^{k}$. Then if we set $x_{2k+1}=x'_{2k+1}-8\eps$, $x_{2k+2}=x'_{2k+2}+8\eps$, the sequence $x_0,\cdots,x_{2K+3}$ satisfies all the properties required for the step.

Identifying each point still takes $O(\log(1/\epsilon)) = \tilde{O}(1)$ queries to $\fs_{\gamma}(\cdot)$ (thus $F(\cdot)$). Given that we have $\tilde{O}(1)$ such points, we used a total of $\tilde{O}(1)$ queries. The last interval $I_{K+1}$ again satisfies the conditions of Lemma \ref{lem:levy-learn}.

\paragraph{Step 2} For each interval $I_k$ with $ k=0,..., K$, partition $I_k$ to intervals $I_{k,1},I_{k,2},\cdots$ of length $\frac{1}{2^k}\sqrt{\eps}$. The number of endpoints added is at most $\frac{|I_k|}{\sqrt{\eps}/2^k}=2^k\eps^{-1/2}|I_k|$ for each interval $I_k$, and sums up to at most
\begin{eqnarray*}
\sum_{k=0}^{K}2^k\eps^{-1/2}|I_k|&=&\eps^{-1/2}|I_0|+\sum_{k=1}^{K}2^k\eps^{-1/2}|I_k|\\
&\leq&\eps^{-1/2}+\sum_{k=1}^{K}\eps^{-1/2}\int_{v\in I_k}4f(v)dv\\
&\leq&\eps^{-1/2}+\eps^{-1/2}\int_{v\in [0,1]}4f(v)dv\\
&=&\eps^{-1/2}+4\eps^{-1/2}=5\eps^{-1/2}.
\end{eqnarray*}
The only change compared to the previous section is that $2^k\leq 4f(v)$ instead of $2f(v)$ in the second line.
Thus if we query $F^*$ for all the endpoints of the intervals from the first two steps, the total number of queries needed is at most $O(\eps^{-1/2})$.

\paragraph{Step 3} For each interval $I_{k,j}=[v_{k,j},v_{k,j+1}]=[v_{k,j},v_{k,j}+2^{-k}\sqrt{\eps}]\subseteq I_k$ partitioned in Step 2, assume that $2^{-k}\sqrt{\eps}>10\eps$, otherwise as $|I_{k,j}|=O(\eps)$, it already satisfies the condition of Lemma~\ref{lem:levy-learn}.

Define $\uf_{k,j}$ and $\of_{k,j}$ as follows. Let $\gamma^*=2^{-k-2}\sqrt{\eps}$. If $I_{k,j}$ is the leftmost interval in $I_k$, then $\uf_{k,j}=2^{k-2}$, otherwise $\uf_{k,j}=\fs_{\gamma^*}(v_{k,j}-2^{-k-1}\sqrt{\eps})-2^{k+3}\sqrt{\eps}$; if $I_{k,j}$ is the rightmost interval in $I_k$, then $\of_{k,j}=2^{k+1}$, otherwise $\of_{k,j}=\fs_{\gamma^*}(v_{k,j}+2^{-k}\sqrt{\eps}+2^{-k-1}\sqrt{\eps})+2^{k+3}\sqrt{\eps}$.
Notice that by \eqref{eqn:pdf-query}, if $\uf_{k,j}\neq 2^{k-2}$, then
\[\fs_{\gamma^*}(v_{k,j}-2^{-k-1}\sqrt{\eps})=(1+\delta_1)f(v_{k,j}-2^{-k-1}\sqrt{\eps}+\gamma_1)+\delta_2\] 
for some $\delta_1,\delta_2,\gamma_1$ where $|\delta_1|\leq\frac{\eps}{2\gamma^*}=2^{k+3}\eps^{1.5}$, $\delta_2\leq\frac{\eps}{\gamma^*}=2^{k+2}\sqrt{\eps}$, $|\gamma_1|<2\gamma=2^{-k-1}\sqrt{\eps}$. This means that $v_{k,j}-2^{-k-1}\sqrt{\eps}+\gamma_1\in[v_{k,j}-2^{-k}\sqrt{\eps},v_{k,j}]$. Since $f(v_{k,j})\leq 2^{K+1}=\frac{2}{\eps}$, we have $f(v_{k,j})|\delta_1|\leq \frac{2}{\eps}\cdot2^{k+1}\eps^{1.5}=2^{k+2}\sqrt{\eps}$. Then
\begin{eqnarray*}
\uf_{k,j}&=&\fs_{\gamma^*}(v_{k,j}-2^{-k-1}\sqrt{\eps})-2^{k+3}\sqrt{\eps}\\
&=&(1+\delta_1)f(v_{k,j}-2^{-k-1}\sqrt{\eps}+\gamma_1)+\delta_2-2^{k+3}\sqrt{\eps}\\
&\leq&(1+\delta_1)f(v_{k,j})+2^{k+2}\sqrt{\eps}-2^{k+3}<f(v_{k,j}).
\end{eqnarray*}
This means that $\uf_{k,j}$ is indeed a lower bound of $f(v_{k,j})$. Symmetrically, $\of_{k,j}$ is indeed an upper bound of $f(v_{k,j+1})$.

If $\overline{f}_{k,j} - \underline{f}_{k,j} =t_{k,j}2^k\sqrt{\eps}$, then partition the interval to $\tilde{t}_{k,j}:=\max(1,\lceil{t_{k,j}}\rceil-2^5)$ intervals of length $\frac{1}{\tilde{t}_{k,j}2^k}\sqrt{\eps}$, and query $F$ for the newly added endpoints.
As any neighboring points $\overline{v}, \underline{v}$ have distance $\overline{v} - \underline{v} \leq \frac{1}{\tilde{t}_{k,j}2^k}\sqrt{\eps}$. Hence this interval satisfies the second condition in Lemma \ref{lem:levy-learn} :
$$ (\overline{v} -  \underline{v}) (\overline{f}_{k,j} -  \underline{f}_{k,j}) \leq \frac{1}{\tilde{t}_{k,j}2^k}\sqrt{\eps}\cdot (\tilde{t}_{k,j}+2^5)2^k\sqrt{\eps}=O(\eps)$$
which allows us to learn the distribution up to $\epsilon$ \levy-distance.

Finally, we only need to bound the number of queries needed in this step. Observe that if $I_{k,j}$ is not the leftmost interval in $I_k$,
\begin{eqnarray*}
\uf_{k,j}&=&(1+\delta_1)f(v_{k,j}-2^{-k-1}\sqrt{\eps}+\gamma_1)+\delta_2-2^{k+3}\sqrt{\eps}\\
&\geq&(1-|\delta_1|)f(v_{k,j}-2^{-k}\sqrt{\eps})-2^{k+2}\sqrt{\eps}-2^{k+3}\sqrt{\eps}\\
&\geq&f(v_{k,j-1})-2^{k+4}\sqrt{\eps}.
\end{eqnarray*}
Symmetrically, if $I_{k,j}$ is not the rightmost interval in $I_k$, $\of_{k,j}\leq f(v_{k,j+2})+2^{k+4}\sqrt{\eps}$. 
Thus, 
\[\of_{k,j}-\uf_{k,j}\leq f(v_{k,j+2})-f(v_{k,j-1})+2^{k+5}\sqrt{\eps}.\]
Since the interval $I_{k,j}$ is partitioned to $\max(1,\lceil{t_{k,j}}\rceil-2^5)$ intervals, the number of queries to $\fs_{\gamma^*}$ is at most $\frac{f(v_{k,j+2})-f(v_{k,j-1})}{2^{k}\sqrt{\eps}}$. If we sum up $\frac{f(v_{k,j+2})-f(v_{k,j-1})}{2^{k}\sqrt{\eps}}$ for every $j$, we get the total number of queries to $\fs_{\gamma^*}$ is at most\footnote{To be more precise, for the rightmost interval in $I_k$, $f(v_{k,j+2})$ is replaced by $2^{k+1}$; for the leftmost interval in $I_k$, $f(v_{k,j-1})$ is replaced by $2^{k-2}$.}
\begin{eqnarray*}
\sum_{j}\frac{f(v_{k,j+2})-f(v_{k,j-1})}{2^{k}\sqrt{\eps}}\leq \frac{3}{2^k\sqrt{\eps}}(2^{k+1}-2^{k-2})=O(\eps^{-0.5})
\end{eqnarray*}
by simplifying the telescoping sum. Thus in this step, for each $k$, the number of queries on $F^*$ for added points is $\tO(\eps^{-1/2})$. Also, the number of queries on $\fs_{\gamma^*}$ for $I_k$ is twice the total number of endpoints in $I_{k,j}$, which is $O(\eps^{-1/2})$ by Step 2. As there are only $\tilde{O}(1)$ different values for $k$, the total number of queries on $F^*$ is $\tO(\eps^{-1/2})$.

To summarize, we have described how to modify the algorithm for learning a distribution with smooth convex CDF with oracle queries to $F$ and $f$, to an algorithm with only oracle queries to $\Fs$. The query complexity is asymptotically the same, which is $\tO(\eps^{-1/2})$.
\end{proof}

\subsection{Query complexity of general regular distributions}

Essentially the same modifications described for the convex CDF can be applied to the general regular distributions, obtaining a proof of Theorem \ref{thm:perturbed-regular}. We omit the details since it is essentially a re-writing of the proof of Theorem~\ref{thm:regular} and instead describe the main modifications.

In Step 0, the partition of $[0,1]$ to intervals with quantile within a factor of 2 can also be done via queries to $\Fs$.
In Step 1 of the proof of Theorem~\ref{thm:regular}, the partition is doable with queries to $f$ replaced by queries to $\fs_{O(\eps)}$, as Lemma~\ref{lem:binary-search-query} only requires estimating $f$ within a constant factor, which is the same as Step 1 for Theorem~\ref{thm:perturbed-convex}. The only difference is that the value range of $f$ in each interval may get expanded by a constant factor. This leads to the number of queries in Step 2 increasing by a constant factor. In the analysis of Step 3, similar to Step 3 of the modified algorithm in Theorem~\ref{thm:perturbed-convex} the accuracy needed for estimating $f$ is $O(2^{k+\ell}\sqrt{\eps})$, which is achievable via queries to $\fs_{O(2^{-k-\ell}\sqrt{\eps})}$.

A simple Chernoff bound shows that a query to $F^*$ (defined in Theorem \ref{thm:perturbed-regular}) can be simulated by $O(1/\epsilon^2)$ pricing queries.
Combining Theorem \ref{thm:perturbed-regular} with Chernoff bound we obtain a proof of Theorem \ref{thm:pricing-query-regular}, matching the lower bound established in \cite{leme2021pricing}.


\section{Mixture Distributions}

Our second application is the analysis of mixtures of regular distributions. Corollaries \ref{cor:mix1} and \ref{cor:mix2} are straightforward from Theorem \ref{thm:main}.

We end the paper with the proof of the more surprising result (Theorem \ref{thm:learning_mixture_regular}) that even though mixtures of $O(1)$ regular distributions can be described with $\tilde{O}(1/\epsilon^{0.5})$ bits, they require $\Omega(1/\epsilon^{3})$ pricing queries to be learned within $\epsilon$ errors, the same asymptotic amount as a general distribution.


\newtheorem*{theorem*}{Theorem}

\begin{theorem*}[Restatement of Theorem \ref{thm:learning_mixture_regular}]
Let $\C$ be the class of distributions supported in $[0,1]$ that can be written as a mixture of two regular distributions. To estimate the CDF of a distribution in that class within $\epsilon$ in \levy-distance we require at least $\Omega(1/\epsilon^3)$ pricing queries.
\end{theorem*}

\begin{proof}[Proof of Theorem \ref{thm:learning_mixture_regular}]
We will use Lemma 3.5 in~\cite{leme2021pricing} which says that it takes  $\Omega(1/\epsilon^2)$ pricing queries to distinguish between any two distributions with CDFs $F_1$ and $F_2$ such that 
\begin{equation}\label{eqn:closedist}
    \frac{1}{1+\epsilon}\le\frac{F_1(v)}{F_2(v)},\frac{1-F_1(v)}{1-F_2(v)}\le 1+\epsilon,\ \forall v.
\end{equation}

Now, we will construct a set $S$ of $\Omega(1/\epsilon)$ distributions such that each of which is: (1) a mixture of two regular distributions and (2) each follows Equation~\ref{eqn:closedist}. Moreover, the \levy-distance between any two pairs is $\Omega(\epsilon)$.

Construct a class of regular distributions parameterized by $a<2,\delta>0$ as follows. For a distribution with PDF $f_{a,\delta}$, 
\begin{equation*}
     f_{a,\delta}(v) = \left\{\begin{array}{lr}
        a, & \text{if } v\in[0,\frac{1-\delta}{a}],\\
        \frac{a^2}{2\delta}(\frac{1+\delta}{a}-v), & \text{if } v\in(\frac{1-\delta}{a},\frac{1+\delta}{a}),\\
        0, & \text{if } v\in[\frac{1+\delta}{a},1].\\
        \end{array}\right. 
\end{equation*}
In other words, $f_a$ is a distribution that is uniform with density $a$, until the point $v=\frac{1-\delta}{a}$ where the CDF reaches $1-\delta$. Then $f_{a,\delta}$ decreases linearly to 0, and remains 0 afterwards. We first show that any distribution in this class is regular. We want to show that for any $x\in[\frac{1-\delta}{a},\frac{1+\delta}{a}]$, $f_{a,\delta}'(x)\geq -\frac{2f^2(x)}{1-F(x)}$. Firstly, in this range, $f_{a,\delta}'(x)=-\frac{a^2}{2\delta}$. Secondly, as in this range $1-F_{a,\delta}(x)=\frac{1}{2}(\frac{1+\delta}{a}-x)f_{a,\delta}(x)$, we have
\[-\frac{2f_{a,\delta}^2(x)}{1-F_{a,\delta}(x)}=-\frac{2f_{a,\delta}(x)}{\frac{1}{2}(\frac{1+\delta}{a}-x)}=-\frac{2\cdot\frac{a^2}{2\delta}(\frac{1+\delta}{a}-x)}{\frac{1}{2}(\frac{1+\delta}{a}-x)}=-\frac{2a^2}{\delta}<-\frac{a^2}{2\delta}=f_{a,\delta}'(x).\]
Thus for any $a,\delta$, $f_{a,\delta}$ is the PDF of a regular distribution. Consider any distribution with PDF $g_{a,\delta}(v)=2-f_{a,\delta}(v)$ for every $v\in[0,1]$. Then such a distribution is also regular, since its PDF is monotone non-decreasing. Furthermore, the uniform mixture of $f_{a,\delta}$ and $g_{a,\delta}$ is a uniform distribution on $[0,1]$.

Let $F$ be the uniform distribution on $[0,1]$.
Define the following class of distributions with CDF $F_a$ and PDF $f_a$ parameterized by $a\in[1.5,1.6]$ as follows. For any $v\in[0,1]$, 
\[f_a(v)=f_{a,\eps/2}(v)+g_{a,\eps}.\]
Then $f_a$ is a uniform distribution, except for $v\in[\frac{1-\eps}{a},\frac{1+\eps}{a}]$.
For any $v\in [\frac{1-\eps}{a},\frac{1+\eps}{a}]$, $F_{a,\eps}(v)$ and $F_{a,\eps/2}(v)$ differs by at most $\eps$, thus $F_{a}(v)=\frac{F_{a,\eps/2}(v)+G_{a,\eps}(v)}{2}$ and $F(v)=\frac{F_{a,\eps}(v)+G_{a,\eps}(v)}{2}$ differs by at most $\eps$. Furthermore, since $F_a(v)$ and $F(v)$ are both constant away from 1, thus \eqref{eqn:closedist} is satisfied and $\Omega(1/\eps^2)$ pricing queries are needed to distinguish $F_a$ and $F$.

Observe that for any $a$, $F_a$ and the uniform distribution $F$ has \levy-distance at least $\Omega(\eps)$, as $F_{a,\eps/2}(\frac{1-\eps/2}{a})-F_{a,\eps}(\frac{1-\eps/2}{a})=\Omega(\eps)$. Take $S=\{F_a|a=1.5+4k\eps\textrm{ for }k=1,2,\cdots,\frac{1}{40\eps}\}$. As any $F_a\in S$ and the uniform distribution $F$ differs only on $[\frac{1-\eps}{a},\frac{1+\eps}{a}]$, we know that $\Omega(1/\eps^{2})$ pricing queries are needed on $[\frac{1-\eps}{a},\frac{1+\eps}{a}]$ to distinguish $F_a$ and $F$. Also since for every $a'<a$, $\frac{1+\eps}{a}<\frac{1-\eps}{a'}$ for $a,a'\in[1.5,1.6]$ and small enough $\eps$, we know that for different $a$, the interval where $F_a$ and $F$ differs are non-overlapping. Thus given any distribution $F_a$ where $a=1.5+4k\eps$ for some $k=1,2,\cdots,\frac{1}{40\eps}$, $\Omega(1/\eps^{3})$ queries are required to learn a distribution within $o(\eps)$ \levy-distance from $F_a$.

\end{proof}

\bibliographystyle{alpha}
\bibliography{refs}

\end{document}